\pdfoutput=1
\documentclass[12pt]{article}
\usepackage{amsmath,amssymb,amsfonts,here}
\usepackage{tikz}
\begin{document}
\newtheorem{theorem}{Theorem}[section]
\newtheorem{corollary}[theorem]{Corollary}
\newtheorem{lemma}[theorem]{Lemma}
\newtheorem{remark}[theorem]{Remark}
\newtheorem{example}[theorem]{Example}
\newtheorem{proposition}[theorem]{Proposition}
\newtheorem{definition}[theorem]{Definition}
\newtheorem{assumption}[theorem]{Assumption}
\def\emptyset{\varnothing}
\def\setminus{\smallsetminus}
\def\id{{\mathrm{id}}}
\def\G{{\mathcal{G}}}
\def\E{{\mathcal{E}}}
\def\H{{\mathcal{H}}}
\def\C{{\mathbb{C}}}
\def\N{{\mathbb{N}}}
\def\Q{{\mathbb{Q}}}
\def\R{{\mathbb{R}}}
\def\Z{{\mathbb{Z}}}
\def\Path{{\mathrm{Path}}}
\def\Str{{\mathrm{Str}}}
\def\st{{\mathrm{st}}}
\def\tr{{\mathrm{tr}}}
\def\opp{{\mathrm{opp}}}
\def\a{{\alpha}}
\def\be{{\beta}}
\def\de{{\delta}}
\def\e{{\varepsilon}}
\def\si{{\sigma}}
\def\la{{\lambda}}
\def\th{{\theta}}
\def\lan{{\langle}}
\def\ran{{\rangle}}
\def\isom{{\cong}}
\newcommand{\Hom}{\mathop{\mathrm{Hom}}\nolimits}
\newcommand{\End}{\mathop{\mathrm{End}}\nolimits}
\def\qed{{\unskip\nobreak\hfil\penalty50
\hskip2em\hbox{}\nobreak\hfil$\square$
\parfillskip=0pt \finalhyphendemerits=0\par}\medskip}
\def\proof{\trivlist \item[\hskip \labelsep{\bf Proof.\ }]}
\def\endproof{\null\hfill\qed\endtrivlist\noindent}

\title{The zipper condition for $4$-tensors\\
in two-dimensional topological order\\
and the higher relative commutants of a subfactor\\
arising from a commuting square}
\author{
{\sc Yasuyuki Kawahigashi}\\
{\small Graduate School of Mathematical Sciences}\\
{\small The University of Tokyo, Komaba, Tokyo, 153-8914, Japan}\\
{\small e-mail: {\tt yasuyuki@ms.u-tokyo.ac.jp}}
\\[0,40cm]
{\small iTHEMS Research Group, RIKEN}\\
{\small 2-1 Hirosawa, Wako, Saitama 351-0198,Japan}
\\[0,05cm]
{\small and}
{\small Kavli IPMU (WPI), the University of Tokyo}\\
{\small 5--1--5 Kashiwanoha, Kashiwa, 277-8583, Japan}}
\maketitle{}
\centerline{\sl Dedicated to the memory of Hikosaburo Komatsu}

\begin{abstract}
Researchers in condensed matter physics recently 
study two-dimensional topological order in terms of
tensor networks involving certain $3$- and $4$-tensors.  
Their $3$-tensors satisfying the ``zipper condition'' 
play an important role there and such $3$-tensors can
be made into certain $2$-tensors by combining two wires
into one.  We
identify their $4$-tensors with bi-unitary connections in 
Jones' subfactor theory in operator algebras with precise
normalization constants.  
Then we prove that their $2$-tensors satisfying the zipper condition
are the same as flat fields of strings in subfactor
theory which correspond to elements in the higher relative
commutants of the subfactor arising from the bi-unitary
connection.  This is what we expect, since the zipper
condition is a kind of pentagon relations, but we
clarify what conditions are exactly needed for this ---
we do not need the flatness or the finite depth condition
for the bi-unitary connection.
We actually generalize their $4$-tensors so that the four
index sets of the $4$-tensors can be all different and
work on a ``half-version'' of the zipper condition.
\end{abstract}

\section{Introduction}

\textit{Fusion categories} \cite{EGNO}  (with or without braiding)
have emerged as new types of symmetries in mathematics
and physics.  Both quantum field theory and condensed
matter physics have seen such symmetries recently and
they are often called \textit{non-invertible
symmetries} in the physics literature, and also called
\text{quantum symmetries} in various mathematics literature.
Particularly, many researchers 
in \textit{two-dimensional topological order} in condensed
matter physics are interested in such studies 
using \textit{tensor networks} recently
as in \cite{BMWSHV}, \cite{MRGSCP}, \cite{RGMP}.

It has been well-known that subfactor theory of Jones
\cite{J1}, \cite{J2} in operator algebras gives appropriate
tools to study structures of fusion
categories, and it is indeed this theory which led to 
the discovery of the \textit{Jones polynomial} of knots
and links, the
first mathematical realization of quantum symmetry.
This approach is closely related to
operator algebraic studies of quantum field theory
\cite{L1}, \cite{L2}, \cite{EK2}, since early days
of subfactor theory through Doplicher-Haag-Roberts
theory in \textit{algebraic quantum field theory}.

In a usual operator algebraic study of fusion categories,
we realize an object of a fusion
category as a bimodule over (type II$_1$)
factors or an endomorphism of a (type III) factor
\cite{EK1}, \cite{EK2}.  Another approach \cite{AH} based on 
bi-unitary connections \cite{O}, \cite{Sc}, \cite{K1} is less
frequently studied, but contains the same information as these
two methods and has an advantage that everything is
\textit{finite dimensional}.  This finite dimensionality
enabled us to construct the Haagerup subfactor \cite{AH},
which is still one of the most mysterious quantum
symmetries today.  Recall that a bi-unitary
connection gives a characterization \cite{Sc},
\cite{EK1} of a non-degenerate \textit{commuting square}, which
was initially studied in \cite{P1} in a different context
and has complete information to recover an amenable 
subfactor of type II$_1$ \cite{P2}.

It has been pointed out
in \cite{K2}, \cite{K4} that the $4$-tensors in 
\cite{BMWSHV} are mathematically the same as
bi-unitary connections up to slight change of
normalization, and identification of
some natural finite dimensional Hilbert spaces in condensed matter 
physics and subfactor theory has been given in \cite{K3}.
A characterization of such $4$-tensors as certain
generalized quantum $6j$-symbols has been also given
in \cite{K5}.  
This shows that \textit{anyons} \cite{Ko}
are studied with such $4$-tensor networks \cite{H}.
Our correspondence among various mathematical
approaches to study fusion categories are
summarized in Tables \ref{corresp} and \ref{corresp2}.

\begin{table}[H]
\centering
\caption{Correspondence among endomorphisms, 
bimodules and connections
\label{corresp}}
\begin{tabular}{|c|c|c|}
\hline
endomorphism & bimodules & connections \\ \hline \hline
identity & identity bimodule & trivial connection \\ \hline
direct sum & direct sum & direct sum \\ \hline
composition &  relative tensor product & composition \\ \hline
conjugate endomorphism & dual bimodule & dual connection \\ \hline
dimension & $(\text{Jones\ index})^{1/2}$ &
Perron-Frobenius eigenvalue \\ \hline
intertwiner & intertwiner & flat field of strings \\ \hline
\end{tabular}
\end{table}

\begin{table}[H]
\centering
\caption{Correspondence between connections, commuting 
squares and $4$-tensors
\label{corresp2}}
\begin{tabular}{|c|c|c|}
\hline
connections & commuting square & $4$-tensor \\ \hline \hline
trivial connection & commuting square & trivial $4$-tensor \\ \hline
direct sum & direct sum& direct sum \\\hline
composition & composition & concatenation \\ \hline
dual connection & basic construction & complex conjugate tensor \\ \hline
Perron-Frobenius eigenvalue & 
$(\text{Pimsner-Popa\ index})^{1/2}$ & Perron-Frobenius eigenvalue \\ \hline
flat fields of strings & relative commutant &
tensors with the zipper condition \\ \hline
\end{tabular}
\end{table}

The aim of this paper is threefold and to complete the
above Tables as follows in our main Theorem \ref{zipper}.

(1) Give precise normalization constants in various formulas.

(2) Characterize the morphism property in terms of the ``if and
only if'' form in comparison to the open string bimodule
framework and the zipper condition \cite[(2)]{BMWSHV}.

(3) Give the most general setting of assumptions under which
our arguments work.

The first one is only a technical issue, but important for
actual computations.  Examples of concrete computations
in various physics papers arise from $3$-cocycles on finite groups,
where all the normalizing constants are $1$ and this issue
can be ignored, but we need them in a more general setting.

The second is important from a theoretical viewpoint.  The
bimodule approach involves infinite dimensional operator
algebras and Hilbert spaces, so it is not clear whether
this approach gives the same morphisms as in the tensor
network framework.

For the third aim, we may have four different index sets
for our $4$-tensors, as long as we have bi-unitarity
as in Fig.~\ref{biuni1} and \ref{biuni2}, and we do not need
the finite depth condition or the flatness condition \cite{EK1}
for bi-unitary connections. The lack of the former condition
means that the initial data can produce 
countably many irreducible objects in our tensor category,
and the lack of the latter condition means that our bi-unitary connection,
a kind of \textit{quantum $6j$-symbols}, does not have to be
in a canonical form.

\section{Preliminaries on bi-unitary connections}

We prepare notations and conventions on bi-unitary connections
as in \cite{AH},
\cite[Chapter 11]{EK1}, \cite{K1}, \cite{O}, \cite{Sc}.
We give complete definitions so that researches in other
fields can understand the setting, since
ours in this paper is slightly more general than the one
in \cite{K2}, \cite{K3}.

We have four finite bipartite oriented
graphs $\G_0,\G_1,\G_2,\G_3$.  We assume that 
$\G_0$ and $\G_2$ are connected.    Each graph can have
multiple edges between one pair of vertices and also cycles,
but cannot have a loop, an edge from one vertex to the same one,
since it is bipartite.  We write $E(\G)$ for
the edge set of a graph $\G$.

We assume that
the sets of the source vertices of $E(\G_0)$,
$E(\G_1)$, $E(\G_2)$ and $E(\G_3)$ are 
$V_0$, $V_0$, $V_1$ and $V_3$, respectively. 
We further assume that the sets of the range vertices of $E(\G_0)$,
$E(\G_1)$, $E(\G_2)$ and $E(\G_3)$ are 
$V_3$, $V_1$, $V_2$ and $V_2$, respectively. 
We draw a diagram as in Fig.~\ref{four} to depict this
situation.  
We assume that the numbers of the edges
of all the four graphs are larger than one.

\begin{figure}[H]
\begin{center}
\begin{tikzpicture}
\draw [thick, ->] (1,1)--(2,1);
\draw [thick, ->] (1,2)--(2,2);
\draw [thick, ->] (1,2)--(1,1);
\draw [thick, ->] (2,2)--(2,1);
\draw (1,1.5)node[left]{$\G_1$};
\draw (2,1.5)node[right]{$\G_3$};
\draw (1.5,1)node[below]{$\G_2$};
\draw (1.5,2)node[above]{$\G_0$};
\draw (1,1)node[below left]{$V_1$};
\draw (1,2)node[above left]{$V_0$};
\draw (2,1)node[below right]{$V_2$};
\draw (2,2)node[above right]{$V_3$};
\end{tikzpicture}
\caption{Four graphs}
\label{four}
\end{center}
\end{figure}

Let $\Delta_{\G_0,xy}$ be the number of edges of $\G_0$ between
$x\in V_0$ and $y\in V_3$.
Let $\Delta_{\G_1,xy}$ be the number of edges of $\G_1$ between
$x\in V_1$ and $y\in V_2$.
Let $\Delta_{\G_2,xy}$ be the number of edges of $\G_2$ between
$x\in V_0$ and $y\in V_1$.
Let $\Delta_{\G_3,xy}$ be the number of edges of $\G_3$ between
$x\in V_3$ and $y\in V_2$.  We assume that we have the following
identities for some positive numbers $\beta_0,\beta_1$.
We assume that we
have a positive number $\mu(x)$ for each vertex $x$
and that the following identities hold.  That is, for each of
$V_0, V_1, V_2, V_3$, the vector given by $\mu(x)$ gives a
Perron-Frobenius eigenvector for the adjacency matrix of one 
of the four graphs, and the numbers $\beta_0,\beta_1$ are
the Perron-Frobenius eigenvalues of these matrices.
Since all the four graphs have more than one edge, we have
$\beta_0,\beta_1 > 1$.  We fix one such $\mu(x)$ for all $x$.

\begin{align*}
\sum_{x\in V_0}\Delta_{\G_0,xy} \mu(x)=\beta_0 (\mu(y)),\quad y\in V_3,\\
\sum_{y\in V_3}\Delta_{\G_0,xy} \mu(y)=\beta_0 (\mu(x)),\quad x\in V_0,\\
\sum_{x\in V_1}\Delta_{\G_2,xy} \mu(x)=\beta_0 (\mu(y)),\quad y\in V_2,\\
\sum_{y\in V_2}\Delta_{\G_2,xy} \mu(y)=\beta_0 (\mu(x)),\quad x\in V_1,\\
\sum_{x\in V_0}\Delta_{\G_1,xy} \mu(x)=\beta_1 (\mu(y)),\quad y\in V_1,\\
\sum_{y\in V_1}\Delta_{\G_1,xy} \mu(y)=\beta_1 (\mu(x)),\quad x\in V_0,\\
\sum_{x\in V_3}\Delta_{\G_3,xy} \mu(x)=\beta_1 (\mu(y)),\quad y\in V_2,\\
\sum_{y\in V_2}\Delta_{\G_3,xy} \mu(y)=\beta_1 (\mu(x)),\quad x\in V_3,
\end{align*}

Here is one example Fig.~\ref{ex1} of four graphs where all the
four graphs are isomorphic and $\beta_0=\beta_1=\sqrt{3}$.

\begin{figure}[H]
\begin{center}
\begin{tikzpicture}
\draw [thick, ->] (1,5)--(2,4);
\draw [thick, ->] (3,5)--(2,4);
\draw [thick, ->] (3,5)--(4,4);
\draw [thick, ->] (5,5)--(4,4);
\draw [thick, ->] (9,5)--(10,4);
\draw [thick, ->] (11,5)--(10,4);
\draw [thick, ->] (11,5)--(12,4);
\draw [thick, ->] (13,5)--(12,4);
\draw [thick, ->] (2,2)--(1,1);
\draw [thick, ->] (2,2)--(3,1);
\draw [thick, ->] (4,2)--(3,1);
\draw [thick, ->] (4,2)--(5,1);
\draw [thick, ->] (10,2)--(9,1);
\draw [thick, ->] (10,2)--(11,1);
\draw [thick, ->] (12,2)--(11,1);
\draw [thick, ->] (12,2)--(13,1);
\draw (3,0.5)node[below]{$\G_2$};
\draw (11,0.5)node[below]{$\G_3$};
\draw (3,3.5)node[below]{$\G_0$};
\draw (11,3.5)node[below]{$\G_1$};
\draw (0.7,1)node[below left]{$V_2$};
\draw (0.7,2)node[above left]{$V_1$};
\draw (0.7,4)node[below left]{$V_3$};
\draw (0.7,5)node[above left]{$V_0$};
\draw (8.7,1)node[below left]{$V_2$};
\draw (8.7,2)node[above left]{$V_3$};
\draw (8.7,4)node[below left]{$V_1$};
\draw (8.7,5)node[above left]{$V_0$};
\draw (1,1)node[below]{$8$};
\draw (3,1)node[below]{$9$};
\draw (5,1)node[below]{$10$};
\draw (9,1)node[below]{$8$};
\draw (11,1)node[below]{$9$};
\draw (13,1)node[below]{$10$};
\draw (2,2)node[above]{$6$};
\draw (4,2)node[above]{$7$};
\draw (10,2)node[above]{$4$};
\draw (12,2)node[above]{$5$};
\draw (2,4)node[below]{$4$};
\draw (4,4)node[below]{$5$};
\draw (10,4)node[below]{$6$};
\draw (12,4)node[below]{$7$};
\draw (1,5)node[above]{$1$};
\draw (3,5)node[above]{$2$};
\draw (5,5)node[above]{$3$};
\draw (9,5)node[above]{$1$};
\draw (11,5)node[above]{$2$};
\draw (13,5)node[above]{$3$};
\end{tikzpicture}
\caption{Example 1:How four graphs are connected}
\label{ex1}
\end{center}
\end{figure}

Another example of four graphs is given in Fig.~\ref{ex2}, where
all the four graphs are different.  We have $\beta_0=2\cos(\pi/12)$
and $\beta_1=(3+\sqrt3)^{1/2}$.

\begin{figure}[H]
\begin{center}
\begin{tikzpicture}
\draw [thick, ->] (1,11)--(2,10);
\draw [thick, ->] (3,11)--(2,10);
\draw [thick, ->] (3,11)--(4,10);
\draw [thick, ->] (5,11)--(4,10);
\draw [thick, ->] (5,11)--(6,10);
\draw [thick, ->] (7,11)--(6,10);
\draw [thick, ->] (7,11)--(8,10);
\draw [thick, ->] (9,11)--(8,10);
\draw [thick, ->] (9,11)--(10,10);
\draw [thick, ->] (11,11)--(10,10);
\draw [thick, ->] (1,8)--(2,7);
\draw [thick, ->] (3,8)--(2,7);
\draw [thick, ->] (3,8)--(4.5,7);
\draw [thick, ->] (4,8)--(4.5,7);
\draw [thick, ->] (5,8)--(4.5,7);
\draw [thick, ->] (6,8)--(4.5,7);
\draw [thick, ->] (6,8)--(7,7);
\draw [thick, ->] (8,8)--(7,7);
\draw [thick, ->] (1,5)--(2,4);
\draw [thick, ->] (3,5)--(2,4);
\draw [thick, ->] (3,5)--(3,4);
\draw [thick, ->] (3,5)--(4,4);
\draw [thick, ->] (5,5)--(4,4);
\draw [thick, ->] (1,2)--(2,1);
\draw [thick, ->] (2,2)--(2,1);
\draw [thick, ->] (2,2)--(3,1);
\draw [thick, ->] (3,2)--(2,1);
\draw [thick, ->] (3,2)--(4,1);
\draw [thick, ->] (4,2)--(3,1);
\draw [thick, ->] (4,2)--(4,1);
\draw [thick, ->] (5,2)--(4,1);
\draw (6,9.5)node[below]{$\G_0$};
\draw (4.5,6.5)node[below]{$\G_1$};
\draw (3,3.5)node[below]{$\G_2$};
\draw (3,0.5)node[below]{$\G_3$};
\draw (0.7,11)node[above left]{$V_0$};
\draw (0.7,10)node[below left]{$V_3$};
\draw (0.7,8)node[above left]{$V_0$};
\draw (0.7,7)node[below left]{$V_1$};
\draw (0.7,5)node[above left]{$V_1$};
\draw (0.7,4)node[below left]{$V_2$};
\draw (0.7,2)node[above left]{$V_3$};
\draw (0.7,1)node[below left]{$V_2$};
\draw (2,1)node[below]{$15$};
\draw (3,1)node[below]{$16$};
\draw (4,1)node[below]{$17$};
\draw (1,2)node[above]{$7$};
\draw (2,2)node[above]{$10$};
\draw (3,2)node[above]{$9$};
\draw (4,2)node[above]{$8$};
\draw (5,2)node[above]{$11$};
\draw (2,4)node[below]{$15$};
\draw (3,4)node[below]{$16$};
\draw (4,4)node[below]{$17$};
\draw (1,5)node[above]{$12$};
\draw (3,5)node[above]{$13$};
\draw (5,5)node[above]{$14$};
\draw (2,7)node[below]{$12$};
\draw (4.5,7)node[below]{$13$};
\draw (7,7)node[below]{$14$};
\draw (1,8)node[above]{$1$};
\draw (3,8)node[above]{$4$};
\draw (4,8)node[above]{$2$};
\draw (5,8)node[above]{$5$};
\draw (6,8)node[above]{$3$};
\draw (8,8)node[above]{$6$};
\draw (2,10)node[below]{$7$};
\draw (4,10)node[below]{$8$};
\draw (6,10)node[below]{$9$};
\draw (8,10)node[below]{$10$};
\draw (10,10)node[below]{$11$};
\draw (1,11)node[above]{$1$};
\draw (3,11)node[above]{$2$};
\draw (5,11)node[above]{$3$};
\draw (7,11)node[above]{$4$};
\draw (9,11)node[above]{$5$};
\draw (11,11)node[above]{$6$};
\end{tikzpicture}
\caption{Example 2:How four graphs are connected}
\label{ex2}
\end{center}
\end{figure}

For an edge $\xi$ of one of the graphs $\G_0,\G_1,\G_2,\G_3$, we
write $s(\xi)$ and $r(\xi)$ for the
source and the range.
Let $\xi_0,\xi_1,\xi_2,\xi_3$ be edges of 
$\G_0,\G_1,\G_2,\G_3$, respectively. 
If we have 
$s(\xi_0)=x_0\in V_0$,
$r(\xi_0)=x_3\in V_3$,
$s(\xi_1)=x_0\in V_0$,
$r(\xi_1)=x_1\in V_1$,
$s(\xi_2)=x_1\in V_1$,
$r(\xi_2)=x_2\in V_2$,
$s(\xi_3)=x_3\in V_3$, and
$r(\xi_3)=x_2\in V_2$,
then we call a combination of $\xi_i$ 
a {\em cell}, as in Fig.~\ref{cell}.

\begin{figure}[H]
\begin{center}
\begin{tikzpicture}
\draw [thick, ->] (1,1)--(2,1);
\draw [thick, ->] (1,2)--(2,2);
\draw [thick, ->] (1,2)--(1,1);
\draw [thick, ->] (2,2)--(2,1);
\draw (1,1.5)node[left]{$\xi_1$};
\draw (2,1.5)node[right]{$\xi_3$};
\draw (1.5,1)node[below]{$\xi_2$};
\draw (1.5,2)node[above]{$\xi_0$};
\draw (1,1)node[below left]{$x_1$};
\draw (1,2)node[above left]{$x_0$};
\draw (2,1)node[below right]{$x_2$};
\draw (2,2)node[above right]{$x_3$};
\end{tikzpicture}
\caption{A cell}
\label{cell}
\end{center}
\end{figure}

For each cell, we assign a complex number.  We call this map
a {\em connection} and write $W$ for this.  We also write
as in Fig.~\ref{connection} for the number assigned by
$W$ to this cell.  If one of the conditions
$s(\xi_0)=s(\xi_1)$,
$r(\xi_0)=s(\xi_3)$,
$r(\xi_1)=s(\xi_2)$ and
$r(\xi_2)=r(\xi_3)$
fails, we understand that the diagram in Fig.~\ref{connection}
denotes the number $0$.

\begin{figure}[H]
\begin{center}
\begin{tikzpicture}
\draw [thick, ->] (1,1)--(2,1);
\draw [thick, ->] (1,2)--(2,2);
\draw [thick, ->] (1,2)--(1,1);
\draw [thick, ->] (2,2)--(2,1);
\draw (1.5,1.5)node{$W$};
\draw (1,1.5)node[left]{$\xi_1$};
\draw (2,1.5)node[right]{$\xi_3$};
\draw (1.5,1)node[below]{$\xi_2$};
\draw (1.5,2)node[above]{$\xi_0$};
\end{tikzpicture}
\caption{A connection value}
\label{connection}
\end{center}
\end{figure}

We first require {\em unitarity} of $W$ 
as in Fig.~\ref{unitarity}, where the bar on the right cell
denotes the complex conjugate.

\begin{figure}[H]
\begin{center}
\begin{tikzpicture}
\draw [thick, ->] (2,1)--(3,1);
\draw [thick, ->] (2,2)--(3,2);
\draw [thick, ->] (2,2)--(2,1);
\draw [thick, ->] (3,2)--(3,1);
\draw (2.5,1.5)node{$W$};
\draw (2,1.5)node[left]{$\xi_1$};
\draw (3,1.5)node[right]{$\xi_3$};
\draw (2.5,1)node[below]{$\xi_2$};
\draw (2.5,2)node[above]{$\xi_4$};
\draw [thick, ->] (4,1)--(5,1);
\draw [thick, ->] (4,2)--(5,2);
\draw [thick, ->] (4,2)--(4,1);
\draw [thick, ->] (5,2)--(5,1);
\draw [thick] (3.7,2.8)--(5.3,2.8);
\draw (4.5,1.5)node{$W$};
\draw (4,1.5)node[left]{$\xi_1$};
\draw (5,1.5)node[right]{$\xi'_3$};
\draw (4.5,1)node[below]{$\xi_2$};
\draw (4.5,2)node[above]{$\xi'_4$};
\draw (0.5,1.5)node{$\sum_{\xi_1,\xi_2}$};
\draw (7.5,1.5)node
{$=\displaystyle\delta_{\xi_3,\xi'_3}\delta_{\xi_4,\xi'_4}$};
\end{tikzpicture}
\caption{Unitarity}
\label{unitarity}
\end{center}
\end{figure}

We define a new connection $W'$ as
in Fig.~\ref{renormalization1}, where $\tilde\xi_0$  
denotes the edge $\xi_0$ with its orientation reversed.
We also require that
this $W'$ satisfies unitarity.  When unitarity holds for
$W$ and $W'$, we say $W$ satisfies {\em bi-unitarity}
and call $W$ a {\em bi-unitary connection}.  Since
we consider only connections with bi-unitarity, we 
simply write a connection for a bi-unitary connection.
Ocneanu and Haagerup found that
a bi-unitary connection characterizes a 
non-degenerate commuting squares of finite dimensional
$C^*$-algebras with a trace as in \cite[Section 11.2]{EK1},

\begin{figure}[H]
\begin{center}
\begin{tikzpicture}
\draw [thick, ->] (2,1)--(3,1);
\draw [thick, ->] (2,2)--(3,2);
\draw [thick, ->] (2,2)--(2,1);
\draw [thick, ->] (3,2)--(3,1);
\draw (2.5,1.5)node{$W'$};
\draw (2,1.5)node[left]{$\xi_3$};
\draw (3,1.5)node[right]{$\xi_1$};
\draw (2.5,1)node[below]{$\tilde\xi_2$};
\draw (2.5,2)node[above]{$\tilde\xi_0$};
\draw [thick, ->] (8,1)--(9,1);
\draw [thick, ->] (8,2)--(9,2);
\draw [thick, ->] (8,2)--(8,1);
\draw [thick, ->] (9,2)--(9,1);
\draw [thick] (7.7,2.8)--(9.3,2.8);
\draw (8.5,1.5)node{$W$};
\draw (8,1.5)node[left]{$\xi_1$};
\draw (9,1.5)node[right]{$\xi_3$};
\draw (8.5,1)node[below]{$\xi_2$};
\draw (8.5,2)node[above]{$\xi_0$};
\draw (5.5,1.5)node{$\displaystyle=\sqrt
{\frac{\mu(s(\xi_0))\mu(r(\xi_2))}{\mu(r(\xi_0))\mu(s(\xi_2))}}$};
\end{tikzpicture}
\caption{Renormalization (1)}
\label{renormalization1}
\end{center}
\end{figure}

\begin{figure}[H]
\begin{center}
\begin{tikzpicture}
\draw [thick, ->] (2,1)--(3,1);
\draw [thick, ->] (2,2)--(3,2);
\draw [thick, ->] (2,2)--(2,1);
\draw [thick, ->] (3,2)--(3,1);
\draw (2.5,1.5)node{$\bar W$};
\draw (2,1.5)node[left]{$\tilde\xi_1$};
\draw (3,1.5)node[right]{$\tilde\xi_3$};
\draw (2.5,1)node[below]{$\xi_0$};
\draw (2.5,2)node[above]{$\xi_2$};
\draw [thick, ->] (8,1)--(9,1);
\draw [thick, ->] (8,2)--(9,2);
\draw [thick, ->] (8,2)--(8,1);
\draw [thick, ->] (9,2)--(9,1);
\draw [thick] (7.7,2.8)--(9.3,2.8);
\draw (8.5,1.5)node{$W$};
\draw (8,1.5)node[left]{$\xi_1$};
\draw (9,1.5)node[right]{$\xi_3$};
\draw (8.5,1)node[below]{$\xi_2$};
\draw (8.5,2)node[above]{$\xi_0$};
\draw (5.5,1.5)node{$\displaystyle=\sqrt
{\frac{\mu(s(\xi_0))\mu(r(\xi_2))}{\mu(r(\xi_0))\mu(s(\xi_2))}}$};
\end{tikzpicture}
\caption{Renormalization (2)}
\label{renormalization2}
\end{center}
\end{figure}

\begin{figure}[H]
\begin{center}
\begin{tikzpicture}
\draw [thick, ->] (2,1)--(3,1);
\draw [thick, ->] (2,2)--(3,2);
\draw [thick, ->] (2,2)--(2,1);
\draw [thick, ->] (3,2)--(3,1);
\draw (2.5,1.5)node{$\bar W'$};
\draw (2,1.5)node[left]{$\tilde\xi_3$};
\draw (3,1.5)node[right]{$\tilde\xi_1$};
\draw (2.5,1)node[below]{$\tilde\xi_0$};
\draw (2.5,2)node[above]{$\tilde\xi_2$};
\draw [thick, ->] (5,1)--(6,1);
\draw [thick, ->] (5,2)--(6,2);
\draw [thick, ->] (5,2)--(5,1);
\draw [thick, ->] (6,2)--(6,1);
\draw (5.5,1.5)node{$W$};
\draw (5,1.5)node[left]{$\xi_1$};
\draw (6,1.5)node[right]{$\xi_3$};
\draw (5.5,1)node[below]{$\xi_2$};
\draw (5.5,2)node[above]{$\xi_0$};
\draw (4,1.5)node{$=$};
\end{tikzpicture}
\caption{Renormalization (3)}
\label{renormalization3}
\end{center}
\end{figure}

We also define new connections $\bar W$ and
$\bar W'$ as in
Fig.~\ref{renormalization2} and \ref{renormalization3}.
They both satisfy bi-unitarity automatically.
We also define a value of another diagram as in
Fig.~\ref{convention}.
Note that we have Fig.~\ref{convention2} due to
Fig.~\ref{renormalization1} and \ref{convention}.

\begin{figure}[H]
\begin{center}
\begin{tikzpicture}
\draw [thick, ->] (3,1)--(2,1);
\draw [thick, ->] (3,2)--(2,2);
\draw [thick, ->] (2,2)--(2,1);
\draw [thick, ->] (3,2)--(3,1);
\draw (2.5,1.5)node{$W$};
\draw (2,1.5)node[left]{$\xi_3$};
\draw (3,1.5)node[right]{$\xi_1$};
\draw (2.5,1)node[below]{$\xi_2$};
\draw (2.5,2)node[above]{$\xi_0$};
\draw [thick, ->] (5,1)--(6,1);
\draw [thick, ->] (5,2)--(6,2);
\draw [thick, ->] (5,2)--(5,1);
\draw [thick, ->] (6,2)--(6,1);
\draw [thick] (4.7,2.8)--(6.3,2.8);
\draw (5.5,1.5)node{$W$};
\draw (5,1.5)node[left]{$\xi_1$};
\draw (6,1.5)node[right]{$\xi_3$};
\draw (5.5,1)node[below]{$\xi_2$};
\draw (5.5,2)node[above]{$\xi_0$};
\draw (4,1.5)node{$=$};
\end{tikzpicture}
\caption{Conjugate convention}
\label{convention}
\end{center}
\end{figure}

\begin{figure}[H]
\begin{center}
\begin{tikzpicture}
\draw [thick, ->] (3,1)--(2,1);
\draw [thick, ->] (3,2)--(2,2);
\draw [thick, ->] (2,2)--(2,1);
\draw [thick, ->] (3,2)--(3,1);
\draw (2.5,1.5)node{$W$};
\draw (2,1.5)node[left]{$\xi_1$};
\draw (3,1.5)node[right]{$\xi_3$};
\draw (2.5,1)node[below]{$\tilde\xi_2$};
\draw (2.5,2)node[above]{$\tilde\xi_0$};
\draw [thick, ->] (8,1)--(9,1);
\draw [thick, ->] (8,2)--(9,2);
\draw [thick, ->] (8,2)--(8,1);
\draw [thick, ->] (9,2)--(9,1);
\draw (8.5,1.5)node{$W$};
\draw (8,1.5)node[left]{$\xi_1$};
\draw (9,1.5)node[right]{$\xi_3$};
\draw (8.5,1)node[below]{$\xi_2$};
\draw (8.5,2)node[above]{$\xi_0$};
\draw (5.5,1.5)node{$\displaystyle=\sqrt
{\frac{\mu(s(\xi_0))\mu(r_(\xi_2))}{\mu(r(\xi_0))\mu(s(\xi_2))}}$};
\end{tikzpicture}
\caption{Renormalization convention}
\label{convention2}
\end{center}
\end{figure}

Suppose we have two connections $W_1$ and $W_2$
as depicted in Fig.~\ref{twoconn}.
Then we define the {\em product connection} as in
Fig.~\ref{product}.  The left hand side is a new product 
connection.  Its top and bottom graphs are $\G_0$ and
$\G_4$, respectively.  Its left graph
is concatenation of $\G_1$ and $\G_5$, and 
its right graphs is concatenation of $\G_3$
and $\G_7$.

\begin{figure}[H]
\begin{center}
\begin{tikzpicture}
\draw [thick, ->] (1,1)--(2,1);
\draw [thick, ->] (1,2)--(2,2);
\draw [thick, ->] (1,2)--(1,1);
\draw [thick, ->] (2,2)--(2,1);
\draw (1.5,1.5)node{$W_1$};
\draw (1,1.5)node[left]{$\G_1$};
\draw (2,1.5)node[right]{$\G_3$};
\draw (1.5,1)node[below]{$\G_2$};
\draw (1.5,2)node[above]{$\G_0$};
\draw (1,1)node[below left]{$V_1$};
\draw (1,2)node[above left]{$V_0$};
\draw (2,1)node[below right]{$V_2$};
\draw (2,2)node[above right]{$V_3$};
\draw [thick, ->] (4,1)--(5,1);
\draw [thick, ->] (4,2)--(5,2);
\draw [thick, ->] (4,2)--(4,1);
\draw [thick, ->] (5,2)--(5,1);
\draw (4.5,1.5)node{$W_2$};
\draw (4,1.5)node[left]{$\G_5$};
\draw (5,1.5)node[right]{$\G_7$};
\draw (4.5,1)node[below]{$\G_4$};
\draw (4.5,2)node[above]{$\G_2$};
\draw (4,1)node[below left]{$V_4$};
\draw (4,2)node[above left]{$V_1$};
\draw (5,1)node[below right]{$V_5$};
\draw (5,2)node[above right]{$V_2$};
\end{tikzpicture}
\caption{Two connections}
\label{twoconn}
\end{center}
\end{figure}

\begin{figure}[H]
\begin{center}
\begin{tikzpicture}
\draw [thick, ->] (1,1)--(2,1);
\draw [thick, ->] (1,3)--(2,3);
\draw [thick, ->] (1,2)--(1,1);
\draw [thick, ->] (2,2)--(2,1);
\draw [thick, ->] (1,3)--(1,2);
\draw [thick, ->] (2,3)--(2,2);
\draw (1,2.5)node[left]{$\xi_1$};
\draw (1,1.5)node[left]{$\xi_5$};
\draw (2,2.5)node[right]{$\xi_3$};
\draw (2,1.5)node[right]{$\xi_7$};
\draw (1.5,1)node[below]{$\xi_4$};
\draw (1.5,3)node[above]{$\xi_0$};
\draw [thick, ->] (5,1.5)--(6,1.5);
\draw [thick, ->] (5,2.5)--(6,2.5);
\draw [thick, ->] (5,2.5)--(5,1.5);
\draw [thick, ->] (6,2.5)--(6,1.5);
\draw (5.5,2)node{$W_1$};
\draw (5,2)node[left]{$\xi_1$};
\draw (6,2)node[right]{$\xi_3$};
\draw (5.5,1.5)node[below]{$\xi_2$};
\draw (5.5,2.5)node[above]{$\xi_0$};
\draw [thick, ->] (7.5,1.5)--(8.5,1.5);
\draw [thick, ->] (7.5,2.5)--(8.5,2.5);
\draw [thick, ->] (7.5,2.5)--(7.5,1.5);
\draw [thick, ->] (8.5,2.5)--(8.5,1.5);
\draw (8,2)node{$W_2$};
\draw (7.5,2)node[left]{$\xi_5$};
\draw (8.5,2)node[right]{$\xi_7$};
\draw (8,1.5)node[below]{$\xi_4$};
\draw (8,2.5)node[above]{$\xi_2$};
\draw (3.5,1.8)node{$=\displaystyle\sum_{\xi_2}$};
\end{tikzpicture}
\caption{The product connection of $W_1$ and $W_2$}
\label{product}
\end{center}
\end{figure}

We now define unitary equivalence of two connections
$W_1$ and $W_2$ on the same graphs depicted as in Fig.~\ref{four}.
Suppose we have two unitary matrices $U$, $V$ whose index sets
are the edge sets of $\G_1$, $\G_3$, respectively.  Furthermore,
we assume $U_{\xi_1,\xi'_1}=0$ if $s(\xi_1)\neq s(\xi'_1)$
or $r(\xi_1)\neq r(\xi'_1)$ and a similar property for $V$.
Then we say $W_1$ and $W_2$ are equivalent if the identity
as in Fig.~\ref{equiv} holds.

\begin{figure}[H]
\begin{center}
\begin{tikzpicture}
\draw [thick, ->] (1,1)--(2,1);
\draw [thick, ->] (1,2)--(2,2);
\draw [thick, ->] (1,2)--(1,1);
\draw [thick, ->] (2,2)--(2,1);
\draw (1.5,1.5)node{$W_1$};
\draw (1,1.5)node[left]{$\xi_1$};
\draw (2,1.5)node[right]{$\xi_3$};
\draw (1.5,1)node[below]{$\xi_2$};
\draw (1.5,2)node[above]{$\xi_0$};
\draw (3.8,1.3)node{$\displaystyle=
\sum_{\xi'_1,\xi'_3}U_{\xi_1,\xi'_1}$};
\draw (7.6,1.45)node{$V_{\xi'_3,\xi_3}$};
\draw [thick, ->] (5.5,1)--(6.5,1);
\draw [thick, ->] (5.5,2)--(6.5,2);
\draw [thick, ->] (5.5,2)--(5.5,1);
\draw [thick, ->] (6.5,2)--(6.5,1);
\draw (6,1.5)node{$W_2$};
\draw (5.5,1.5)node[left]{$\xi'_1$};
\draw (6.5,1.5)node[right]{$\xi'_3$};
\draw (6,1)node[below]{$\xi_2$};
\draw (6,2)node[above]{$\xi_0$};
\end{tikzpicture}
\caption{Unitary equivalence of $W_1$ and $W_2$}
\label{equiv}
\end{center}
\end{figure}

We now assume that we have two connections $W_1$ and $W_2$
as in Fig.~\ref{twoconn2}.

\begin{figure}[H]
\begin{center}
\begin{tikzpicture}
\draw [thick, ->] (1,1)--(2,1);
\draw [thick, ->] (1,2)--(2,2);
\draw [thick, ->] (1,2)--(1,1);
\draw [thick, ->] (2,2)--(2,1);
\draw (1.5,1.5)node{$W_1$};
\draw (1,1.5)node[left]{$\G_1$};
\draw (2,1.5)node[right]{$\G_3$};
\draw (1.5,1)node[below]{$\G_2$};
\draw (1.5,2)node[above]{$\G_0$};
\draw (1,1)node[below left]{$V_1$};
\draw (1,2)node[above left]{$V_0$};
\draw (2,1)node[below right]{$V_2$};
\draw (2,2)node[above right]{$V_3$};
\draw [thick, ->] (4,1)--(5,1);
\draw [thick, ->] (4,2)--(5,2);
\draw [thick, ->] (4,2)--(4,1);
\draw [thick, ->] (5,2)--(5,1);
\draw (4.5,1.5)node{$W_2$};
\draw (4,1.5)node[left]{$\G'_1$};
\draw (5,1.5)node[right]{$\G'_3$};
\draw (4.5,1)node[below]{$\G_2$};
\draw (4.5,2)node[above]{$\G_0$};
\draw (4,1)node[below left]{$V_1$};
\draw (4,2)node[above left]{$V_0$};
\draw (5,1)node[below right]{$V_2$};
\draw (5,2)node[above right]{$V_3$};
\end{tikzpicture}
\caption{Two connections}
\label{twoconn2}
\end{center}
\end{figure}

We define the sum graph $\G''_1$ of $\G_1$ and $\G'_1$
as follows.  This is a bipartite
graph with the two disjoint vertex sets $V_0$ and $V_1$ and the
edge set being the disjoint union of 
$E(\G_1)$ and $E(\G'_1)$.  We similarly define the sum 
graph $\G''_3$ of $\G_3$ and $\G'_3$
We next define the direct sum connection $W_\oplus W_2$ as in
Fig.\ref{direct} on the four graphs in Fig.~\ref{four2}.

\begin{figure}[H]
\begin{center}
\begin{tikzpicture}
\draw [thick, ->] (1,1)--(2,1);
\draw [thick, ->] (1,2)--(2,2);
\draw [thick, ->] (1,2)--(1,1);
\draw [thick, ->] (2,2)--(2,1);
\draw (1,1.5)node[left]{$\G''_1$};
\draw (2,1.5)node[right]{$\G''_3$};
\draw (1.5,1)node[below]{$\G_2$};
\draw (1.5,2)node[above]{$\G_0$};
\draw (1,1)node[below left]{$V_1$};
\draw (1,2)node[above left]{$V_0$};
\draw (2,1)node[below right]{$V_2$};
\draw (2,2)node[above right]{$V_3$};
\end{tikzpicture}
\caption{The four graphs for $W_1\oplus W_2$}
\label{four2}
\end{center}
\end{figure}

\begin{figure}[H]
\begin{center}
\begin{tikzpicture}
\draw [thick, ->] (1,2.5)--(2,2.5);
\draw [thick, ->] (1,3.5)--(2,3.5);
\draw [thick, ->] (1,3.5)--(1,2.5);
\draw [thick, ->] (2,3.5)--(2,2.5);
\draw (1,3)node[left]{$\xi''_1$};
\draw (2,3)node[right]{$\xi''_3$};
\draw (1.5,2.5)node[below]{$\xi_2$};
\draw (1.5,3.5)node[above]{$\xi_0$};
\draw (3.5,3.1)node{$=\left\{\begin{array}{c}
\vphantom{X}\\
\vphantom{X}\\
\vphantom{X}\\
\vphantom{X}\\
\vphantom{X}\\
\vphantom{X}\\
\vphantom{X}\\
\vphantom{X}\\
\vphantom{X}\\
\vphantom{X}\\
\vphantom{X}\\
\vphantom{X}
\end{array}\right.$};
\draw [thick, ->] (4.5,2.5)--(5.5,2.5);
\draw [thick, ->] (4.5,3.5)--(5.5,3.5);
\draw [thick, ->] (4.5,5)--(5.5,5);
\draw [thick, ->] (4.5,6)--(5.5,6);
\draw [thick, ->] (4.5,3.5)--(4.5,2.5);
\draw [thick, ->] (5.5,3.5)--(5.5,2.5);
\draw [thick, ->] (4.5,6)--(4.5,5);
\draw [thick, ->] (5.5,6)--(5.5,5);
\draw (5,5.5)node{$W_1$};
\draw (5,3)node{$W_2$};
\draw (4.5,3)node[left]{$\xi''_1$};
\draw (5.5,5.5)node[right]{$\xi'''_3$};
\draw (4.5,5.5)node[left]{$\xi''_1$};
\draw (5.5,3)node[right]{$\xi''_3$};
\draw (5,2.5)node[below]{$\xi_2$};
\draw (5,3.5)node[above]{$\xi_0$};
\draw (5,5)node[below]{$\xi_2$};
\draw (5,6)node[above]{$\xi_0$};
\draw (5,0.8)node{$0$};
\draw (10,5.5)node{if $\xi''_1\in E(\G_1),\xi''_3\in E(\G_3)$,};
\draw (10,3)node{if $\xi''_1\in E(\G'_1),\xi''_3\in E(\G'_3)$,};
\draw (8.65,0.8)node{otherwise.};
\end{tikzpicture}
\caption{The direct sum connection}
\label{direct}
\end{center}
\end{figure}

If a connection $W$ is written as $W_1\oplus W_2$, we call
it a {\em direct sum decomposition}.
For a connection $W$, if none of the connection
unitarily equivalent to $W$ has a direct sum decomposition,
we say that $W$ is {\em irreducible}.

At the end of this Section, we present how a {\em flat field of strings}
in the sense of \cite[page 563]{EK1} acts on an open string
bimodule in \cite[Claim 1 on page 19]{AH}, because this action
is not explicitly written in \cite{AH}.

We start with a connection $W$ as in the above. Then we choose
initial vertices $*_0\in V_0$ and $*_1\in V_1$, and construct
an open string bimodule $X^W$ as in \cite[page 14]{AH}.
Take a general element in $X^W$ before the completion as in
Fig.~\ref{openstring}.  Without loss of generality, we may
assume that the horizontal length of this string is $2k$, even, 
and thus $x\in V_0$.

\begin{figure}[H]
\begin{center}
\begin{tikzpicture}
\draw [thick, ->] (1,2)--(2,2);
\draw [thick, ->] (2,2)--(3,2);
\draw [thick, ->] (4,2)--(5,2);
\draw [thick, ->] (5,2)--(5,1);
\draw [thick, ->] (1,1)--(2,1);
\draw [thick, ->] (2,1)--(3,1);
\draw [thick, ->] (4,1)--(5,1);
\draw (1,2)node[left]{$*_0$};
\draw (1,1)node[left]{$*_1$};
\draw (5,2)node[right]{$x$};
\draw (5,1)node[right]{$y$};
\draw (3.7,1)node{$\cdots$};
\draw (3.7,2)node{$\cdots$};
\end{tikzpicture}
\caption{An element of $X^W$}
\label{openstring}
\end{center}
\end{figure}

Let $f$ be a flat field of strings on $\G_1$.  This means
the following. Label the
edges on $\G_1$ as $\rho_1,\rho_2,\dots,\rho_n$.
Then a flat field $f$ is written as in Fig.~\ref{field}
which satisfies Fig.~\ref{flatf}.

\begin{figure}[H]
\begin{center}
\begin{tikzpicture}
\draw [thick, ->] (3.7,1.5)--(3.7,0.5);
\draw [thick, ->] (5.2,1.5)--(5.2,0.5);
\draw (1.4,0.8)node{$\displaystyle
\sum_{s(\rho_i)=s(\rho_j),r(\rho_i)=r(\rho_j)} f_{\rho_i,\rho_j}($};
\draw (3.8,1)node[right]{$\rho_i$};
\draw (5.3,1)node[right]{$\rho_j$};
\draw (4.55,0.8)node{$,$};
\draw (5.95,1)node{$)$};
\end{tikzpicture}
\caption{A flat field $f$}
\label{field}
\end{center}
\end{figure}

\begin{figure}[H]
\begin{center}
\begin{tikzpicture}
\draw [thick, ->] (3,1)--(4,1);
\draw [thick, ->] (4,1)--(5,1);
\draw [thick, ->] (3,2)--(4,2);
\draw [thick, ->] (4,2)--(5,2);
\draw [thick, ->] (3,3)--(4,3);
\draw [thick, ->] (4,3)--(5,3);
\draw [thick, ->] (3,1)--(3,2);
\draw [thick, ->] (3,3)--(3,2);
\draw [thick, ->] (4,1)--(4,2);
\draw [thick, ->] (4,3)--(4,2);
\draw [thick, ->] (5,1)--(5,2);
\draw [thick, ->] (5,3)--(5,2);
\draw (3.5,2.5)node{$W_a$};
\draw (3.5,1.5)node{$W_{\bar a}$};
\draw (4.5,2.5)node{$W_{a'}$};
\draw (4.5,1.5)node{$W_{\bar a'}$};
\draw (3.5,1)node[below]{$\xi'_1$};
\draw (3.5,3)node[above]{$\xi_1$};
\draw (4.5,1)node[below]{$\tilde\xi'_2$};
\draw (4.5,3)node[above]{$\tilde\xi_2$};
\draw (3,1.5)node[left]{$\rho_2$};
\draw (3,2.5)node[left]{$\rho_1$};
\draw (5,1.5)node[right]{$\rho'_2$};
\draw (5,2.5)node[right]{$\rho'_1$};
\draw (1.8,1.8)node{$\displaystyle\sum_{\rho_1,\rho_2}f_{\rho_1,\rho_2}$};
\draw (7,2)node{$=\delta_{\xi_1,\xi_1'}\delta_{\xi_2,\xi_2'}
f_{\rho'_1,\rho'_2}$};
\end{tikzpicture}
\caption{Flatness of $f$}
\label{flatf}
\end{center}
\end{figure}

Now fix the vertices $x\in V_0$ and $y\in V_1$.
Label edges from $x$ to $y$ on $\G_1$ as $\rho_1,\rho_2,\dots,\rho_n$.
Then the part of $f$ starting with $x$ and ending with $y$
is written as in Fig.~\ref{string}.

\begin{figure}[H]
\begin{center}
\begin{tikzpicture}
\draw [thick, ->] (2.5,1.5)--(2.5,0.5);
\draw [thick, ->] (4,1.5)--(4,0.5);
\draw (1.4,0.8)node{$\displaystyle\sum_{i,j} f_{\rho_i,\rho_j}($};
\draw (2.5,1)node[right]{$\rho_i$};
\draw (4,1)node[right]{$\rho_j$};
\draw (3.25,0.8)node{$,$};
\draw (4.75,1)node{$)$};
\end{tikzpicture}
\caption{A part of a flat field $f$}
\label{string}
\end{center}
\end{figure}

We now define an action of $f$ on the element in Fig.~\ref{openstring}.
We may assume that the edge from $x$ to $y$ in Fig.~\ref{openstring}
is $\rho_k$.  Then the result of this action is defined to be as in
Fig.~\ref{openstring2}.

\begin{figure}[H]
\begin{center}
\begin{tikzpicture}
\draw [thick, ->] (2,2)--(3,2);
\draw [thick, ->] (3,2)--(4,2);
\draw [thick, ->] (5,2)--(6,2);
\draw [thick, ->] (6,2)--(6,1);
\draw [thick, ->] (2,1)--(3,1);
\draw [thick, ->] (3,1)--(4,1);
\draw [thick, ->] (5,1)--(6,1);
\draw (2,2)node[left]{$*_0$};
\draw (2,1)node[left]{$*_1$};
\draw (6,2)node[right]{$x$};
\draw (6,1)node[right]{$y$};
\draw (4.7,1)node{$\cdots$};
\draw (4.7,2)node{$\cdots$};
\draw (6,1.5)node[right]{$\rho_i$};
\draw (0.7,1.4)node{$\displaystyle\sum_i f_{\rho_i,\rho_k}$};
\end{tikzpicture}
\caption{The result of an action of $f$}
\label{openstring2}
\end{center}
\end{figure}

We prove that this action is well-defined.
Consider a vector represented by Fig.~\ref{openstring} and
write $s$ for this.  We also label the edge from $x$ to $y$
in Fig.~\ref{openstring} as $\rho_i$.
We rewrite the vector $s$ using the basis corresponding to the
diagram in Fig.~\ref{diag}.  This element is represented as
in Fig.~\ref{openstring3}.  We now consider the action of $f$ on this
element.  Due to the well-definedness of the action of string algebras
on the open string bimodule, this action is given by the action of
$f$ written in terms of the basis corresponding to the diagram
in Fig.~\ref{diag2}, but this is simply the action of
a parallel transport $f'$ of $f$ on the vector $s$ written
in terms of basis as in Fig.~\ref{diag}.
(See \cite[Definition 11.18]{EK1} for the notion of
parallel transport.)  Then this is exactly equal to the action of
$f$ on $s$ considered with respect to the basis corresponding
to the diagram in Fig.~\ref{diag2}.  The same argument
shows the well-definedness of the step from $2k$ to
$2k+1$, so we have proved the following Proposition.

\begin{figure}[H]
\begin{center}
\begin{tikzpicture}
\draw [thick, ->] (1,2)--(2,2);
\draw [thick, ->] (2,2)--(3,2);
\draw [thick, ->] (4,2)--(5,2);
\draw [thick, ->] (5,2)--(6,2);
\draw [thick, ->] (6,2)--(6,1);
\draw [thick, ->] (1,1)--(2,1);
\draw [thick, ->] (2,1)--(3,1);
\draw [thick, ->] (4,1)--(5,1);
\draw [thick, ->] (5,1)--(6,1);
\draw (1,2)node[left]{$*_0$};
\draw (1,1)node[left]{$*_1$};
\draw (5,2)node[above]{$2k$};
\draw (6,2)node[above]{$2k+1$};
\draw (3.7,1)node{$\cdots$};
\draw (3.7,2)node{$\cdots$};
\end{tikzpicture}
\caption{A basis for a finite dimensional subspace of $X^W$}
\label{diag}
\end{center}
\end{figure}

\begin{figure}[H]
\begin{center}
\begin{tikzpicture}
\draw [thick, ->] (4,2)--(5,2);
\draw [thick, ->] (5,2)--(6,2);
\draw [thick, ->] (7,2)--(8,2);
\draw [thick, ->] (8,2)--(9,2);
\draw [thick, ->] (9,2)--(9,1);
\draw [thick, ->] (4,1)--(5,1);
\draw [thick, ->] (5,1)--(6,1);
\draw [thick, ->] (7,1)--(8,1);
\draw [thick, ->] (8,1)--(9,1);
\draw (4,2)node[left]{$*_0$};
\draw (4,1)node[left]{$*_1$};
\draw (8,2)node[above]{$2k$};
\draw (9,2)node[above]{$2k+1$};
\draw (8.5,2)node[below]{$\eta'$};
\draw (9,1.5)node[right]{$\xi'$};
\draw (8.5,1)node[below]{$\eta$};
\draw (6.7,1)node{$\cdots$};
\draw (6.7,2)node{$\cdots$};
\draw [thick, ->] (2,2)--(3,2);
\draw [thick, ->] (2,1)--(3,1);
\draw [thick, ->] (3,2)--(3,1);
\draw [thick, ->] (2,2)--(2,1);
\draw (2.5,1.5)node{$W$};
\draw (2.5,2)node[above]{$\eta'$};
\draw (3,1.5)node[right]{$\xi'$};
\draw (2.5,1)node[below]{$\eta$};
\draw (2,1.5)node[left]{$\xi_i$};
\draw (1.1,1.3)node{$\displaystyle\sum_{\eta,\eta,\xi'}$};
\end{tikzpicture}
\caption{The element $s$ written in terms of a new basis}
\label{openstring3}
\end{center}
\end{figure}

\begin{figure}[H]
\begin{center}
\begin{tikzpicture}
\draw [thick, ->] (1,2)--(2,2);
\draw [thick, ->] (2,2)--(3,2);
\draw [thick, ->] (4,2)--(5,2);
\draw [thick, ->] (5,2)--(6,2);
\draw [thick, ->] (6,2)--(6,1);
\draw (1,2)node[left]{$*_0$};
\draw (5,2)node[above]{$2k$};
\draw (6,2)node[above]{$2k+1$};
\draw (3.7,2)node{$\cdots$};
\end{tikzpicture}
\caption{A basis for a finite dimensional string subalgebra}
\label{diag2}
\end{center}
\end{figure}

\begin{proposition}
The above action of flat fields of strings gives a 
self-intertwiner of $X^W$ commuting with the left
and right actions of hyperfinite II$_1$ factors
arising from the string algebras on $\mathcal G_0$
and $\mathcal G_2$, and all self-intertwiners of
$X^W$ arise in this way.
\end{proposition}

\section{$4$-tensors}

We start with a connection $W_a$ and define the corresponding
$4$-tensor $a$ as in Fig.~\ref{tensor}.
We have the conjugate connection $\overline{W_a}$ as in
Fig.~\ref{renormalization2}.  We also write 
$W_{\bar a}$ for this connection and define
a $4$-tensor $\bar a$ corresponding to $W_{\bar a}$
as in Fig.~\ref{tensor2}.

\begin{figure}[H]
\begin{center}
\begin{tikzpicture}
\draw [thick] (2.5,1)--(2.5,1.2);
\draw [thick] (2.5,2)--(2.5,1.8);
\draw [thick] (2,1.5)--(2.2,1.5);
\draw [thick] (3,1.5)--(2.8,1.5);
\draw (2.5,1.5) circle (0.3);
\draw (2.5,1.5)node{$a$};
\draw (2,1.5)node[left]{$\rho$};
\draw (3,1.5)node[right]{$\sigma$};
\draw (2.5,1)node[below]{$\eta$};
\draw (2.5,2)node[above]{$\xi$};
\draw [thick, ->] (8,1)--(9,1);
\draw [thick, ->] (8,2)--(9,2);
\draw [thick, ->] (8,2)--(8,1);
\draw [thick, ->] (9,2)--(9,1);
\draw (8.5,1.5)node{$W_a$};
\draw (8,1.5)node[left]{$\rho$};
\draw (8.5,1)node[below]{$\eta$};
\draw (9,1.5)node[right]{$\sigma$};
\draw (8.5,2)node[above]{$\xi$};
\draw (5.5,1.5)node{$\displaystyle=
\sqrt[4]{\frac{\mu(s(\xi))\mu(r(\eta))}{\mu(r(\xi))\mu(s(\eta))}}$};
\end{tikzpicture}
\caption{The $4$-tensor $a$ and the connection $W_a$}
\label{tensor}
\end{center}
\end{figure}

\begin{figure}[H]
\begin{center}
\begin{tikzpicture}
\draw [thick] (2.5,1)--(2.5,1.2);
\draw [thick] (2.5,2)--(2.5,1.8);
\draw [thick] (2,1.5)--(2.2,1.5);
\draw [thick] (3,1.5)--(2.8,1.5);
\draw (2.5,1.5) circle (0.3);
\draw (2.5,1.5)node{$\bar a$};
\draw (2,1.5)node[left]{$\tilde\rho$};
\draw (3,1.5)node[right]{$\tilde\sigma$};
\draw (2.5,1)node[below]{$\xi$};
\draw (2.5,2)node[above]{$\eta$};
\draw [thick] (5.5,1)--(5.5,1.2);
\draw [thick] (5.5,2)--(5.5,1.8);
\draw [thick] (5,1.5)--(5.2,1.5);
\draw [thick] (6,1.5)--(5.8,1.5);
\draw (5.5,1.5) circle (0.3);
\draw (5.5,1.5)node{$a$};
\draw (5,1.5)node[left]{$\rho$};
\draw (6,1.5)node[right]{$\sigma$};
\draw (5.5,1)node[below]{$\eta$};
\draw (5.5,2)node[above]{$\xi$};
\draw [thick] (4.8,2.8)--(6.2,2.8);
\draw (4,1.5)node{$=$};
\draw [thick, ->] (11,1)--(12,1);
\draw [thick, ->] (11,2)--(12,2);
\draw [thick, ->] (11,2)--(11,1);
\draw [thick, ->] (12,2)--(12,1);
\draw (11.5,1.5)node{$W_{\bar a}$};
\draw (11,1.5)node[left]{$\tilde\rho$};
\draw (11.5,1)node[below]{$\xi$};
\draw (12,1.5)node[right]{$\tilde\sigma$};
\draw (11.5,2)node[above]{$\eta$};
\draw (8.5,1.5)node{$\displaystyle=
\sqrt[4]{\frac{\mu(r(\xi))\mu(s(\eta))}{\mu(s(\xi))\mu(r(\eta))}}$};
\end{tikzpicture}
\caption{The $4$-tensors $a, \bar a$ and 
the connection $W_{\bar a}$}
\label{tensor2}
\end{center}
\end{figure}

Then we have Figures \ref{biuni1} and \ref{biuni2} to represent
bi-unitarity.

\begin{figure}[H]
\begin{center}
\begin{tikzpicture}
\draw [thick] (4.5,1)--(4.5,1.2);
\draw [thick] (4.5,2)--(4.5,1.8);
\draw [thick] (4,1.5)--(4.2,1.5);
\draw [thick] (5,1.5)--(4.8,1.5);
\draw (4.5,1.5) circle (0.3);
\draw (4.5,1.5)node{$\bar a$};
\draw [thick] (4.5,2)--(4.5,2.2);
\draw [thick] (4.5,3)--(4.5,2.8);
\draw [thick] (4,2.5)--(4.2,2.5);
\draw [thick] (5,2.5)--(4.8,2.5);
\draw (4.5,2.5) circle (0.3);
\draw (4.5,2.5)node{$a$};
\draw [thick] (4,2.5) arc (90:270:0.5);
\draw (3.6,2.3)node[above left]{$\rho$};
\draw (3.6,1.8)node[below left]{$\tilde\rho$};
\draw (5,1.5)node[right]{$\tilde\sigma'$};
\draw (5,2.5)node[right]{$\sigma$};
\draw (4.5,1)node[below]{$\xi'$};
\draw (4.5,3)node[above]{$\xi$};
\draw (4.5,2)node[right]{$\eta$};
\draw (1,2)node{$\displaystyle\sum_{\eta,\rho}
\sqrt{\frac{\mu(r(\xi))\mu(s(\eta))}
{\mu(s(\xi))\mu(r(\eta))}}$};
\draw (7.5,2)node{$\displaystyle=
\delta_{\xi,\xi'}\delta_{\sigma,\sigma'}
\delta_{r(\xi),s(\sigma)}$};
\end{tikzpicture}
\caption{Bi-unitarity (1)}
\label{biuni1}
\end{center}
\end{figure}

\begin{figure}[H]
\begin{center}
\begin{tikzpicture}
\draw [thick] (3.5,1)--(3.5,1.2);
\draw [thick] (3.5,2)--(3.5,1.8);
\draw [thick] (3,1.5)--(3.2,1.5);
\draw [thick] (4,1.5)--(3.8,1.5);
\draw (3.5,1.5) circle (0.3);
\draw (3.5,1.5)node{$\bar a$};
\draw [thick] (3.5,2)--(3.5,2.2);
\draw [thick] (3.5,3)--(3.5,2.8);
\draw [thick] (3,2.5)--(3.2,2.5);
\draw [thick] (4,2.5)--(3.8,2.5);
\draw (3.5,2.5) circle (0.3);
\draw (3.5,2.5)node{$a$};
\draw [thick] (4.5,2) arc (0:90:0.5);
\draw [thick] (4,1.5) arc (270:360:0.5);
\draw (4.3,2.3)node[above right]{$\sigma$};
\draw (4.3,1.8)node[below right]{$\tilde\sigma$};
\draw (3,1.5)node[left]{$\tilde\rho'$};
\draw (3,2.5)node[left]{$\rho$};
\draw (3.5,1)node[below]{$\xi'$};
\draw (3.5,3)node[above]{$\xi$};
\draw (3.5,2)node[right]{$\eta$};
\draw (0.5,2)node{$\displaystyle\sum_{\eta,\sigma}
\sqrt{\frac{\mu(s(\xi))\mu(r(\eta))}{
\mu(r(\xi))\mu(s(\eta))}}$};
\draw (6.5,2)node{$\displaystyle=
\delta_{\xi,\xi'}\delta_{\rho,\rho'}
\delta_{s(\xi),s(\rho)}$};
\end{tikzpicture}
\caption{Bi-unitarity (2)}
\label{biuni2}
\end{center}
\end{figure}

We can simply represent bi-unitarity in the two identity
as in a diagram in Fig.\ref{biuni3}, where we drop all labels
and the Kronecker $\delta$'s.

\begin{figure}[H]
\begin{center}
\begin{tikzpicture}
\draw [thick] (1.5,1)--(1.5,1.2);
\draw [thick] (1.5,2)--(1.5,1.8);
\draw [thick] (1,1.5)--(1.2,1.5);
\draw [thick] (2,1.5)--(1.8,1.5);
\draw (1.5,1.5) circle (0.3);
\draw [thick] (1.5,2)--(1.5,2.2);
\draw [thick] (1.5,3)--(1.5,2.8);
\draw [thick] (1,2.5)--(1.2,2.5);
\draw [thick] (2,2.5)--(1.8,2.5);
\draw (1.5,2.5) circle (0.3);
\draw [thick] (1,2.5) arc (90:270:0.5);
\draw (2.8,2)node{$=$};
\draw [thick] (4.5,1)--(4.5,1.2);
\draw [thick] (4.5,3)--(4.5,2.8);
\draw [thick] (5,1.5)--(4.8,1.5);
\draw [thick] (5,2.5)--(4.8,2.5);
\draw [thick] (4,1.5)--(4.2,1.5);
\draw [thick] (4,2.5)--(4.2,2.5);
\draw [thick] (4.5,1.8)--(4.5,2.2);
\draw [thick] (4,2.5) arc (90:270:0.5);
\draw [thick] (4.8,2.5) arc (90:180:0.3);
\draw [thick] (4.2,2.5) arc (270:360:0.3);
\draw [thick] (4.5,1.8) arc (180:270:0.3);
\draw [thick] (4.5,1.2) arc (0:90:0.3);
\draw [thick] (6.5,1)--(6.5,1.2);
\draw [thick] (6.5,2)--(6.5,1.8);
\draw [thick] (6,1.5)--(6.2,1.5);
\draw [thick] (7,1.5)--(6.8,1.5);
\draw (6.5,1.5) circle (0.3);
\draw [thick] (6.5,2)--(6.5,2.2);
\draw [thick] (6.5,3)--(6.5,2.8);
\draw [thick] (6,2.5)--(6.2,2.5);
\draw [thick] (7,2.5)--(6.8,2.5);
\draw (6.5,2.5) circle (0.3);
\draw [thick] (7.5,2) arc (0:90:0.5);
\draw [thick] (7,1.5) arc (270:360:0.5);
\draw [thick] (9.5,1)--(9.5,1.2);
\draw [thick] (9.5,2)--(9.5,1.8);
\draw [thick] (9,1.5)--(9.2,1.5);
\draw [thick] (10,1.5)--(9.8,1.5);
\draw [thick] (9.5,2)--(9.5,2.2);
\draw [thick] (9.5,3)--(9.5,2.8);
\draw [thick] (9,2.5)--(9.2,2.5);
\draw [thick] (10,2.5)--(9.8,2.5);
\draw [thick] (10.5,2) arc (0:90:0.5);
\draw [thick] (10,1.5) arc (270:360:0.5);
\draw [thick] (9.5,2.8) arc (180:270:0.3);
\draw [thick] (9.5,2.2) arc (0:90:0.3);
\draw [thick] (9.2,1.5) arc (270:360:0.3);
\draw [thick] (9.8,1.5) arc (90:180:0.3);
\draw (5.4,1.6)node{$,$};
\draw (8.2,2)node{$=$};
\end{tikzpicture}
\caption{Graphical representation of bi-unitarity in a simplified form}
\label{biuni3}
\end{center}
\end{figure}

\section{Flatness and the zipper condition}

We start with a connection $W_a$ and the corresponding
$4$-tensor $a$ as in Fig.~\ref{tensor}.  We work on
$2$-tensors $F,F'$ of the types as in Fig.~\ref{twotensor},
where the index set for $\rho_1,\rho_2$ is the same as the
one for $\rho$ in $a$, and the one for $\sigma_1,\sigma_2$
is the same as the one for $\sigma$ in $a$.

\begin{figure}[H]
\begin{center}
\begin{tikzpicture}
\draw [thick] (3.3,1.2)--(3.7,1.2);
\draw [thick] (3.3,1.8)--(3.7,1.8);
\draw [thick] (3.3,1.2)--(3.3,1.8);
\draw [thick] (3.7,1.2)--(3.7,1.8);
\draw [thick] (3,1.5)--(3.3,1.5);
\draw [thick] (4,1.5)--(3.7,1.5);
\draw (3.5,1.5)node{$F$};
\draw (3,1.5)node[above]{$\rho_2$};
\draw (4,1.5)node[above]{$\rho_1$};
\draw [thick] (5.3,1.2)--(5.7,1.2);
\draw [thick] (5.3,1.8)--(5.7,1.8);
\draw [thick] (5.3,1.2)--(5.3,1.8);
\draw [thick] (5.7,1.2)--(5.7,1.8);
\draw [thick] (5,1.5)--(5.3,1.5);
\draw [thick] (6,1.5)--(5.7,1.5);
\draw (5.5,1.5)node{$F'$};
\draw (5,1.5)node[above]{$\sigma_2$};
\draw (6,1.5)node[above]{$\sigma_1$};
\end{tikzpicture}
\caption{The $2$-tensors $F$  and $F'$}
\label{twotensor}
\end{center}
\end{figure}

Suppose we have a field of strings $f$ which we do not assume
to be flat yet.  Then we define a $2$-tensor $F$ from $f$ as 
in Fig.~\ref{fieldf}.

\begin{figure}[H]
\begin{center}
\begin{tikzpicture}
\draw [thick] (3.3,1.2)--(3.7,1.2);
\draw [thick] (3.3,1.8)--(3.7,1.8);
\draw [thick] (3.3,1.2)--(3.3,1.8);
\draw [thick] (3.7,1.2)--(3.7,1.8);
\draw [thick] (3,1.5)--(3.3,1.5);
\draw [thick] (4,1.5)--(3.7,1.5);
\draw (3.5,1.5)node{$F$};
\draw (3,1.5)node[above]{$\rho_2$};
\draw (4,1.5)node[above]{$\rho_1$};
\draw (5.8,1.5)node{$=\displaystyle
\frac{\mu(r(\rho_1))}{\mu(s(\rho_1))}
f_{\rho_1,\rho_2}$};
\end{tikzpicture}
\caption{The $2$-tensor $F$ arising from 
the field $\displaystyle\sum_{\rho_1,\rho_2}
f_{\rho_1,\rho_2}(\rho_1,\rho_2)$ of strings}
\label{fieldf}
\end{center}
\end{figure}

We now have the following main Theorem.

\begin{theorem}\label{zipper}
The following are equivalent for a $2$-tensor $F$ and 
the corresponding field $f$ of strings defined as above.

\textnormal{(1) [The half zipper condition]}
There exists another $2$ tensor $\tilde F$ so that
the $2$-tensors $F, \tilde F$ satisfy the intertwining property 
as in Fig.~\ref{intertwine1}.

\textnormal{(2) [The zipper condition]}
The $2$-tensors $F$ satisfies the invariance property 
as in Fig.~\ref{intertwine2}.

\textnormal{(3) [Half flatness]}
There exists another field $\tilde f$ of strings so that we have
the half flatness as in Fig.~\ref{halfflatf}.

\textnormal{(4) [Flatness]}
The field $f$ of strings satisfies
the flatness as in Fig.~\ref{flatf}.
\end{theorem}

\begin{proof}

We first show equivalence of (3) and (4)
Recall that this has been essentially
proved in \cite[pages 563--564]{EK1}, but we give more
details in the current context.
Using the initial connection $W_a$, we apply the
string algebra construction in \cite[Section 11.3]{EK1},
but allow all vertices in $V_0$ to be starting vertices.
We then have a double sequence $\{A_{jk}\}_{j,k=1,2,\dots}$
of finite dimensional $C^*$-algebras and $A_{00}$
is an abelian algebra $\mathbb{C}^{|V_0|}$, where
$|V_0|$ denotes the cardinality of $V_0$.  

We assume (3).  The fields of strings $f$ and $\tilde f$
give the corresponding same elements in the algebras $A_{10}$
and $A_{11}$.  We use the symbol $f$ for this.  
Half flatness implies $f$ commutes with $A_{01}$.
The first horizontal Jones projection $e_1$ commutes with
$A_{10}$, so it commutes with $f$, in particular.  This
means $f$ commutes with $A_{02}$ which is generated by
$A_{01}$ and $e_1$.  This shows $f$ produces another
field of strings $\bar f$ on $\mathcal{G}_1$.  The
argument for $z=z'$ in \cite[Fig.~11.16]{EK1} shows
that the field of strings $\bar f$ is equal to the
field of strings $f$.  This implies (4).

Conversely, we assume (4).  In the same way as the
above argument, we construct string algebras 
$\{A_{jk}\}_{j,k=1,2,\dots}$.  The flatness of
$f$ shows that $f$ gives an element in $A_{10}$ which
commutes with $A_{02}$.  In particular, it commutes with
$A_{01}$, and produces a field $\tilde f$ of strings
satisfying the half flatness condition.

We next prove that (3) implies (1).
We first assume half flatness for $f$ and $\tilde f$.
We define a $2$-tensor $\tilde F$ from $\tilde f$ in
a similar way to the definition of $F$.

\begin{figure}[H]
\begin{center}
\begin{tikzpicture}
\draw [thick, ->] (3,1)--(4,1);
\draw [thick, ->] (3,2)--(4,2);
\draw [thick, ->] (3,3)--(4,3);
\draw [thick, ->] (3,1)--(3,2);
\draw [thick, ->] (3,3)--(3,2);
\draw [thick, ->] (4,1)--(4,2);
\draw [thick, ->] (4,3)--(4,2);
\draw (3.5,2.5)node{$W_a$};
\draw (3.5,1.5)node{$W_{\bar a}$};
\draw (3.5,1)node[below]{$\xi'$};
\draw (3.5,3)node[above]{$\xi$};
\draw (3,1.5)node[left]{$\rho_2$};
\draw (3,2.5)node[left]{$\rho_1$};
\draw (4,1.5)node[right]{$\sigma_2$};
\draw (4,2.5)node[right]{$\sigma_1$};
\draw (1.8,1.8)node{$\displaystyle\sum_{\rho_1,\rho_2}f_{\rho_1,\rho_2}$};
\draw (5.5,2)node{$=\delta_{\xi,\xi'}\tilde f_{\sigma_1,\sigma_2}$};
\end{tikzpicture}
\caption{Half flatness of $f$}
\label{halfflatf}
\end{center}
\end{figure}

\begin{figure}[H]
\begin{center}
\begin{tikzpicture}
\draw [thick] (8,1)--(8,1.2);
\draw [thick] (8,2)--(8,1.8);
\draw [thick] (6.8,1.5)--(7.7,1.5);
\draw [thick] (8.5,1.5)--(8.3,1.5);
\draw (8,1.5) circle (0.3);
\draw (8,1.5)node{$\bar a$};
\draw [thick] (8,2)--(8,2.2);
\draw [thick] (8,3)--(8,2.8);
\draw [thick] (7.2,2.5)--(7.7,2.5);
\draw [thick] (8.5,2.5)--(8.3,2.5);
\draw (8,2.5) circle (0.3);
\draw (8,2.5)node{$a$};
\draw [thick] (6.8,2.2)--(7.2,2.2);
\draw [thick] (6.8,2.8)--(7.2,2.8);
\draw [thick] (6.8,2.2)--(6.8,2.8);
\draw [thick] (7.2,2.2)--(7.2,2.8);
\draw (7,2.5)node{$F$};
\draw [thick] (6.8,2.5) arc (90:270:0.5);
\draw (7.5,2.5)node[above]{$\rho_1$};
\draw (7.5,1.5)node[below]{$\tilde\rho_2$};
\draw (8.5,1.5)node[right]{$\tilde\sigma_2$};
\draw (8.5,2.5)node[right]{$\sigma_1$};
\draw (8,1)node[below]{$\xi'$};
\draw (8,3)node[above]{$\xi$};
\draw (8,2)node[right]{$\eta$};
\draw (3.1,1.9)node{$\displaystyle\sum_{\eta,\rho_1,\rho_2}
\sqrt{\frac{\mu(s(\xi))\sqrt{\mu(r(\xi))\mu(r(\xi'))}}
{\mu(s(\eta))\mu(r(\eta))}}$};
\draw [thick] (12.8,1.7)--(13.2,1.7);
\draw [thick] (12.8,2.3)--(13.2,2.3);
\draw [thick] (12.8,1.7)--(12.8,2.3);
\draw [thick] (13.2,1.7)--(13.2,2.3);
\draw [thick] (12.8,2)--(12.3,2);
\draw [thick] (13.2,2)--(13.7,2);
\draw (13,2)node{$\tilde F$};
\draw (13.55,2)node[above]{$\sigma_1$};
\draw (12.45,2)node[above]{$\sigma_2$};
\draw (10.6,2)node{$=\delta_{\xi,\xi'}\displaystyle
\frac{\mu(s(\sigma_1))}{\mu(r(\sigma_1))}$};
\end{tikzpicture}
\caption{Half flatness for $F$}
\label{halfflatff}
\end{center}
\end{figure}

We rewrite the coefficient within the summation on the
left-hand side as
\[
\frac{\mu(s(\xi))}{\mu(s(\eta))}
\sqrt{\frac{\sqrt{\mu(r(\xi))\mu(r(\xi'))}\mu(s(\eta))}
{\mu(s(\xi))\mu(r(\eta))}},
\]
and multiply the number in Fig.~\ref{multi} to the numbers on
the both hands sides of Fig.~\ref{halfflatff} and sum them
over $\xi',\sigma_2$.  

\begin{figure}[H]
\begin{center}
\begin{tikzpicture}
\draw [thick] (3,1)--(3,1.2);
\draw [thick] (3,2)--(3,1.8);
\draw [thick] (2.5,1.5)--(2.7,1.5);
\draw [thick] (3.5,1.5)--(3.3,1.5);
\draw (3,1.5) circle (0.3);
\draw (3,1.5)node{$a$};
\draw (2.5,1.5)node[left]{$\rho_3$};
\draw (3.5,1.5)node[right]{$\sigma_2$};
\draw (3,1)node[below]{$\eta'$};
\draw (3,2)node[above]{$\xi'$};
\draw (0.2,1.5)node{$\displaystyle\sqrt[4]{\frac{\mu(r(\xi'))\mu(s(\eta'))}
{\mu(s(\xi'))\mu(r(\eta'))}}$};
\end{tikzpicture}
\caption{The multiplier}
\label{multi}
\end{center}
\end{figure}

Then by bi-unitarity (2), only the terms for
$\eta=\eta'$ and $\rho_2=\rho_3$ remain, and
we have the identity as in
Fig.~\ref{intertwine1} by dividing the both
hand sides by 
\[
\frac{\mu(r(\sigma_1))}{\mu(s(\sigma_1))}
\displaystyle\sqrt[4]{\frac{\mu(r(\xi))\mu(s(\eta'))}
{\mu(s(\xi))\mu(r(\eta'))}}.
\]
Note that only the term $\xi=\xi'$ remains on
the right hand side due to $\delta_{\xi,\xi'}$.
This proves (1).

Since the above graphical manipulation amounts to
a multiplication of a unitary matrix, we also have
the converse direction.  That is, 
we know that (1) implies (3).

A similar argument to the proof of equivalence of (1)
and (3) shows equivalence of (2) and (4).

\begin{figure}[H]
\begin{center}
\begin{tikzpicture}
\draw [thick] (2.5,1.5)--(2.5,1.7);
\draw [thick] (2.5,2.5)--(2.5,2.3);
\draw [thick] (1.7,2)--(2.2,2);
\draw [thick] (3,2)--(2.8,2);
\draw (2.5,2) circle (0.3);
\draw (2.5,2)node{$a$};
\draw [thick] (1.3,1.7)--(1.7,1.7);
\draw [thick] (1.3,2.3)--(1.7,2.3);
\draw [thick] (1.3,1.7)--(1.3,2.3);
\draw [thick] (1.7,1.7)--(1.7,2.3);
\draw [thick] (1.3,2)--(1,2);
\draw (1.5,2)node{$F$};
\draw (1,2)node[left]{$\rho_3$};
\draw (3,2)node[right]{$\sigma_1$};
\draw (2.5,2.5)node[above]{$\xi$};
\draw (2.5,1.5)node[below]{$\eta'$};
\draw (4,2)node{$=$};
\draw [thick] (5.5,1.5)--(5.5,1.7);
\draw [thick] (5.5,2.5)--(5.5,2.3);
\draw [thick] (5,2)--(5.2,2);
\draw [thick] (6.3,2)--(5.8,2);
\draw (5.5,2) circle (0.3);
\draw (5.5,2)node{$a$};
\draw [thick] (6.3,1.7)--(6.7,1.7);
\draw [thick] (6.3,2.3)--(6.7,2.3);
\draw [thick] (6.3,1.7)--(6.3,2.3);
\draw [thick] (6.7,1.7)--(6.7,2.3);
\draw [thick] (6.7,2)--(7,2);
\draw (6.5,2)node{$\tilde F$};
\draw (5,2)node[left]{$\rho_3$};
\draw (7,2)node[right]{$\sigma_1$};
\draw (5.5,2.5)node[above]{$\xi$};
\draw (5.5,1.5)node[below]{$\eta'$};
\end{tikzpicture}
\caption{Intertwining property for $F, \tilde F$}
\label{intertwine1}
\end{center}
\end{figure}

\begin{figure}[H]
\begin{center}
\begin{tikzpicture}
\draw [thick] (2.5,1.5)--(2.5,1.7);
\draw [thick] (2.5,2.5)--(2.5,2.3);
\draw [thick] (3.5,1.5)--(3.5,1.7);
\draw [thick] (3.5,2.5)--(3.5,2.3);
\draw [thick] (1.7,2)--(2.2,2);
\draw [thick] (2.8,2)--(3.2,2);
\draw [thick] (3.8,2)--(4,2);
\draw (2.5,2) circle (0.3);
\draw (2.5,2)node{$a$};
\draw (3.5,2) circle (0.3);
\draw (3.5,2.07)node{$a'$};
\draw [thick] (1.3,1.7)--(1.7,1.7);
\draw [thick] (1.3,2.3)--(1.7,2.3);
\draw [thick] (1.3,1.7)--(1.3,2.3);
\draw [thick] (1.7,1.7)--(1.7,2.3);
\draw [thick] (1.3,2)--(1,2);
\draw (1.5,2)node{$F$};
\draw (1,2)node[left]{$\rho_1$};
\draw (4,2)node[right]{$\rho_2$};
\draw (2.5,2.5)node[above]{$\xi_1$};
\draw (3.5,2.5)node[above]{$\tilde\xi_2$};
\draw (2.5,1.5)node[below]{$\eta_1$};
\draw (3.5,1.5)node[below]{$\tilde\eta_2$};
\draw (5,2)node{$=$};
\draw [thick] (6.5,1.5)--(6.5,1.7);
\draw [thick] (6.5,2.5)--(6.5,2.3);
\draw [thick] (7.5,1.5)--(7.5,1.7);
\draw [thick] (7.5,2.5)--(7.5,2.3);
\draw [thick] (6,2)--(6.2,2);
\draw [thick] (6.8,2)--(7.2,2);
\draw [thick] (7.8,2)--(8.3,2);
\draw [thick] (8.7,2)--(9,2);
\draw (6.5,2) circle (0.3);
\draw (7.5,2) circle (0.3);
\draw (6.5,2)node{$a$};
\draw (7.5,2.07)node{$a'$};
\draw [thick] (8.3,1.7)--(8.7,1.7);
\draw [thick] (8.3,2.3)--(8.7,2.3);
\draw [thick] (8.3,1.7)--(8.3,2.3);
\draw [thick] (8.7,1.7)--(8.7,2.3);
\draw (8.5,2)node{$F$};
\draw (6,2)node[left]{$\rho_1$};
\draw (9,2)node[right]{$\rho_2$};
\draw (6.5,2.5)node[above]{$\xi_1$};
\draw (6.5,1.5)node[below]{$\eta_1$};
\draw (7.5,2.5)node[above]{$\tilde\xi_2$};
\draw (7.5,1.5)node[below]{$\tilde\eta_2$};
\end{tikzpicture}
\caption{Intertwining property for $F$}
\label{intertwine2}
\end{center}
\end{figure}

\end{proof}

\section*{Acknowledgements}

This work was supported by 
Japan Science and Technology Agency (JST) as 
CREST program JPMJCR18T6, and Grants-in-Aid 
for Scientific Research 24K21514 and 25H00590.
A part of this work was done in Simons Laufer Mathematical Sciences
Institute in Berkeley, Leipzig University and
Fordham University.  The author acknowledges supports of
the National Science Foundation under Grant No.~DMS-1928930 and 
U.S. Army Research Office through contract W911NF-25-1-0050.

\end{document}